\documentclass[pra,amsmath,amssymb,amsfonts,nofootinbib]{revtex4}

\usepackage{bm,graphicx,mathrsfs}

\usepackage{amsmath,bbm}
\usepackage{amsfonts,amssymb}
\usepackage{graphicx}
\usepackage{epsfig}
\usepackage{times}
\usepackage{verbatim}

\newcommand{\R}{\mathbbm{R}}
\newcommand{\rr}{\mathbbm{R}}

\newcommand{\cc}{\mathbbm{C}}

\newcommand{\1}{\mathbbm{1}}

\newcommand{\HS}{{\mathfrak{H}}}

\newcommand{\lio}{{\cal L}}
\newcommand{\M}{{\cal M}}

\newcommand{\K}{{\cal K}}
\newcommand{\sym}{{\cal S}}

\newcommand{\qed}{}
\newcommand{\tr}{{\rm tr}}

\newcommand{\be}{\begin{equation}}
\newcommand{\ee}{\end{equation}}
\newcommand{\bea}{\begin{eqnarray}}
\newcommand{\eea}{\end{eqnarray}}

\newcommand{\ket}[1]{|#1\rangle}
\newcommand{\bra}[1]{\langle#1|}

\newcommand{\avr}[1]{\langle#1\rangle}

\newcommand{\ra}{\rightarrow}

\def\Proof{\noindent\textsc{Proof:}}
\def\proof{\Proof}
\def\qed{\leavevmode\unskip\penalty9999 \hbox{}\nobreak\hfill
     \quad\hbox{\leavevmode  \hbox to.77778em{%
               \hfil\vrule   \vbox to.675em%
               {\hrule width.6em\vfil\hrule}\vrule\hfil}}
     \par\vskip3pt}

\begin{document}

\newtheorem{theorem}{Theorem}
\newtheorem{lemma}[theorem]{Lemma}
\newtheorem{corollary}[theorem]{Corollary}
\newtheorem{proposition}[theorem]{Proposition}
\newtheorem{definition}[theorem]{Definition}
\newtheorem{example}[theorem]{Example}
\newtheorem{conjecture}[theorem]{Conjecture}

\newenvironment{remark}{\vspace{1.5ex}\par\noindent{\it Remark}}%
    {\hspace*{\fill}$\Box$\vspace{1.5ex}\par}

\centerline{\sc \large The $\chi^2$ - divergence and Mixing times of quantum Markov processes}

\author{K. Temme$^1$}
\author{M. J. Kastoryano$^{2,3}$}
\author{M. B. Ruskai$^4$}
\author{M. M. Wolf$^2$}
\author{F. Verstraete$^1$}

\affiliation{$^1$Faculty of Physics, University of Vienna, Boltzmanngasse 5, A-1090 Vienna, Austria\\
		$^2$Niels Bohr Institute, Blegdamsvej 17, DK-2100 Copenhagen, Denmark \\
		$^3$Niels Bohr International Academy, Blegdamsvej 17, 
		DK-2100 Copenhagen, Denmark \\
		$^4$Department of Mathematics, Tufts University, Medford, MA 02155, USA}

\date{\today}

\begin{abstract}
We introduce quantum versions of the $\chi^2$-divergence, provide a
detailed analysis of their properties, and apply them in the
investigation of mixing times of quantum Markov processes. An approach
similar to the one presented in \cite{Diaconis,Fill,Mihail} for classical Markov chains is
taken to bound the trace-distance from the steady state of a quantum
processes. A strict spectral bound to the convergence rate can be
given for time-discrete as well as for time-continuous quantum Markov
processes. Furthermore the contractive behavior of the
$\chi^2$-divergence under the action of a completely positive map is
investigated and contrasted to the contraction of the trace norm. In
this context we analyse different versions of quantum detailed balance
and, finally, give a geometric conductance bound to the convergence
rate for unital quantum Markov processes.
\end{abstract}

\maketitle

\section{Introduction}\label{intro}

The mixing time of a classical Markov chain is the time it takes for the chain to be close to its steady state distribution, 
starting from an arbitrary initial state. The ability to bound the mixing time is important, for example in the field 
of computer science, where the bound can be used to give an estimate for the running time of some probabilistic algorithm 
such as the Monte Carlo algorithm. The mixing time for a classical Markov process  $P_{ij}$, with $\sum_i P_{ij} = 1$ on 
the space of probability measures $\sym$ is commonly defined in terms of the one norm, 
$\| p \|_1 = \sum_i |p_i|$. Let $ \pi$ denote the fixed point of the classical Markov  process, i.e. $P\pi = \pi$, 
then the mixing time is defined as:
\be
	t_{mix}(\epsilon) = \min \left\{ n \; | \; \forall q \in \sym  \; , \; \| P^n \; q - \pi \|_1 < \epsilon \right\}.
\ee
A large set of tools has emerged over the years that allows to investigate the convergence rate of classical Markov chains \cite{Mtime}. 
One of the most prominent approaches \cite{Diaconis,Fill,Mihail} to bounding the mixing time of a Markov chain is based on 
the $\chi^2$-divergence \cite{Pears:chi2}. This divergence is defined for two probability distributions $p, q \in \sym$ as, 
\be \label{classchi2} 
\chi^2(p,q)=\sum_i\frac{(p_i-q_i)^2}{q_i}. 
\ee
The usefulness of the $\chi^2$-divergence for finding bounds to the mixing time of classical Markov 
chains arises from the fact that it serves as an upper bound to the one norm difference between 
two probability distributions, \be \| p  -  q \|_1^2 \leq \chi^2(p,q)\ee and allows for an easier access
to the spectral properties of the Markov chain. The  $\chi^2$-divergence is intimately related to the 
Kullback-Leibler divergence, or relative entropy,  $H(p,q)=\sum_i p_i(\log{p_i}-\log{q_i})$. In fact, it can be obtained 
directly from the relative entropy as the approximating quadratic form, i.e. as the Hessian, of the latter: 
\be
\chi^2(p,q) = \left . -\frac{\partial^2}{\partial_\alpha \partial_\beta}H(q + \alpha(p - q) ,q + \beta(p - q)) \right .|_{\alpha = \beta = 0}. \label{chientcl}
\ee
The $\chi^2$ divergence was first introduced by Karl Pearson in the context of statistical inference tests, the most widely 
used of which is the "Pearson's $\chi^2$ test". Its computational simplicity and its clear relation to other distance 
measures have made it one of the most studied divergence measures in the literature.

In this paper, we find convergence bounds for arbitrary quantum Markov chains, also called quantum channels, 
with a technique that can be seen as a generalization of  the work of \cite{Diaconis,Fill,Mihail} to non-commutative probability spaces.
A prototypical example of mixing time in physics is the decoherence time of the underlying quantum process, i.e. the time in which quantum states 
decohere to an (often classical) mixture given a specific underlying noise model.  The ability to bound the mixing time for quantum 
processes also turns out to be relevant when one seeks to give bounds on the runtime of quantum algorithms 
that are based on quantum Markov chains \cite{dissFrank,quantMet}. Other applications of such bounds can be found in 
the framework of matrix product states \cite{fcsFannes,MPS}, where the correlation length of the quantum state is
connected to the convergence of the corresponding transfer operator that can be interpreted as a quantum channel.
In this article, we introduce the mathematical framework necessary to extend the classical mixing time results to the quantum setting. 
In particular, we introduce a new divergence measure - the quantum $\chi^2$-divergence - for quantum states and use it to obtain some 
basic convergence bounds that mirror existing classical ones. Furthermore, we extend the classical concept of detailed balance to the 
quantum setting and discuss its relevance in general terms. \\

The paper is organized as follows; 
The remainder of section \ref{intro} is devoted to setting the notation and to recalling the framework of quantum channels. 
Then in section \ref{chi2}, we introduce the quantum $\chi^2$-divergence, and prove some basic properties relating it to other 
divergence and distance measures. In particular we focus on a specific subfamily of interest. In section \ref{contraction}, 
we consider contraction of the $\chi^2$-divergence under the action of a channel, and relate it to trace-norm contraction. 
 Furthermore, we prove some fundamental quantum mixing time results, whose classical analogues are well known.  
 In section \ref{detailedbalance}, we study quantum detailed balance, and in section \ref{cheeger}, 
 we extend an important classical geometric mixing bound (Cheeger's inequality) to the quantum setting.
 Conclusions are drawn in \ref{concl}.

\subsection{Formal setting and notation}
Throughout this paper we will consider linear maps from the complex $d$-dimensional 
matrix algebra $\M_d$ to itself. The states are density matrices $\rho \in \sym_d$,
where $\sym_d = \left \{ \rho \in \M_d | \rho = \rho^{\dagger}, \rho \geq 0, \tr[\rho] = 1 \right \}$,
acting on $\HS = \cc^d$. The set of pure states is denoted $\sym^1_d$, while the set of positive definite  
states is denoted $\sym_d^+$. Note, that $\M_d$ itself  becomes a Hilbert space when equipped 
with the standard  Hilbert-Schmidt scalar product $\langle A| B \rangle \equiv \tr[A^\dagger B]$; 
this Hilbert space is naturally isomorphic to $M_d \simeq \cc^{d^2}$. The eigenvalues and 
singular values of $T$ are understood in terms of the matrix representation $\hat{T} \in \M_{d^2}$
of $T$ acting on $\cc^{d^2}$. The matrix representation of a quantum operation, which will 
always be written with a hat (ex. $\hat{T}$),  is given in terms of some 
complete orthonormal basis $\{ F_{i} \}_{i=1 \ldots d^2}$ of $\M_d$, where its matrix 
elements are $\hat{T}_{ij} = \langle F_{i} | T |  F_{j} \rangle$. Unless otherwise specified, we consider 
the basis of matrix units. The distance between states
$\rho_1,\rho_2 \in S_d$ will be measured in terms of the trace distance $\| \rho_1 - \rho_2 \|_1 $ 
induced by the trace norm $\| A \|_1 = \tr[\sqrt{A^{\dagger}A}]=\sum_i s_i(A)$ for $s_i(A)$ the singular 
values of $A \in \M_d$.  Time discrete quantum Markovian dynamics are described by
completely positive, trace-preserving maps (cpt-maps, or quantum channels) \cite{Holevo} $T :\M_d \mapsto \M_d$. 
Due to the Kraus representation theorem, every cp-map can be expressed in terms of the Kraus operators $A_\mu \in M_d$,  as
\be T(\rho) = \sum_\mu A_\mu \rho A_\mu^\dagger \;\;\; \mbox{with,} \;\;\; \sum_\mu A_\mu^{\dagger}A_\mu = \1, \ee
The dual map $T^*$ can be seen as the hermitian conjugate of $T$ with respect to the Hilbert-Schmidt scalar product.
In the above matrix representation of the map, this corresponds exactly to taking the hermitian conjugate $\hat{T}^\dagger$.  A quantum channel
is called unital, or doubly stochastic, when the dual map $T^*$ is also trace preserving. This immediately implies that
$T$ has a fixed point $\sigma = \1/\tr{[\1]}$. We will also consider time continuous quantum Markov
processes, described by a one-parameter semi-group \be T_t(\rho)=e^{t\lio}(\rho)\ee
The Liouvillian $\lio$, obeys $\lio^*(\1) = 0$ \cite{Lindblad}, and its matrix representation $\hat{\lio}$ is obtained as in the time-discrete case. 
In this article, we shall often consider \textit{primitive} quantum channels, i.e. channels with a unique maximal rank fixed point, 
and a unique eigenvalue of magnitude $1$ \cite{Sanz}.

%
% Section: The Quantum chi^2-divergence 
%
\section{The quantum $\chi^2$-divergence}\label{chi2}

We want to define a generalization of the classical $\chi^2$-divergence to the case when we are working on spaces 
with non-commuting density matrices. We shall require that any generalization to the setting of density matrices 
satisfies the condition that when the inputs are diagonal, the classical $\chi^2$-divergence is recovered. 
The first observation we make, reading straight off from (\ref{classchi2}), is that the classical $\chi^2$-divergence can be seen as an inner product on the probability 
space weighted with the inversion of the distribution $q_i$. Due to the non-commutative nature of density matrices there 
is no unique generalization of this inversion. Consider for instance a generalization for two density matrices $\rho,\sigma \in S_d$,
where for now we assume $\sigma$ to be full rank,  that is given by 
\bea\label{alpha-fam}
\chi^2_\alpha(\rho,\sigma) = \tr \left [(\rho-\sigma)\sigma^{-\alpha} (\rho-\sigma)\sigma^{\alpha-1} \right ] 
= \tr \left[ \rho \sigma^{-\alpha} \rho \sigma^{\alpha -1} \right] - 1.
\eea 
This gives rise to an entire family of $\chi^2$-divergences with (as we see below) special properties, for every $\alpha \in [0,1]$. 
The natural question of whether there exists a classification of all possible inversions of $\sigma$, 
was investigated in a series of papers by Morozova and Chentsov \cite{Morozova} Petz \cite{Petz1,Petz2,Petz3}, in the context of information geometry. 
They considered the characterization of monotone Riemannian metrics on matrix spaces. 
Their general definition is based on the modular operator formalism of Araki \cite{Araki1,Araki2}, 
which we will also consider here. In order to classify the valid inversions, we first need to define the 
following set of functions, each of which gives rise to a possible inversion: 
\be
 \K = \{k|-k \mbox{ is operator monotone, } k(w^{-1})=wk(w), \mbox{ and } k(1)=1\}.
\ee

Now, we define left and right multiplication operators as $L_Y(X)=YX$ and $R_Y(X)=XY$ respectively. The \textit{modular operator} is defined as 
\be
	 \Delta_{\rho,\sigma}=L_\rho R^{-1}_\sigma,
\ee 
for all $\rho,\sigma \in \sym_d$, $\sigma >0$. Note, that $R_\sigma$ and $L_\rho$ commute and inherit hermicity and positivity from $\rho,\sigma$. The above should be read as follows: acting on some $A \in M_d$, $\Delta_{\rho,\sigma}(A)=\rho A\sigma^{-1}$. When manipulating the modular operators it is often convenient to write them in matrix form, in wich case, they read: $\hat{\Delta}_{\rho,\sigma}\ket{A} = \rho \otimes \overline{\sigma}^{-1} \ket{A}$, where $\ket{A} = A \otimes \1 \ket{I}$, and $\ket{I} = \sum_{i=1}^d \ket{ii}$ corresponds to $d$ times the maximally mixed state. This formalism gives rise to a more general quantum $\chi^2$-divergence. 
\begin{definition}
For  $\rho,\sigma \in \sym_d$, and $k \in \K$ we define the the quantum $\chi^2$-divergence 
\be \label{chi2def} \chi^2_k(\rho,\sigma)=\left \langle \rho-\sigma,\Omega^k_\sigma(\rho-\sigma) \right \rangle, \ee
when $supp(\rho) \subseteq supp(\sigma)$, and infinity otherwise. The inversion of $\sigma$ is defined only when $supp(\rho) \subseteq supp(\sigma)$, and given by
\be \label{genInv} \Omega^k_\sigma = R^{-1}_{\sigma}k(\Delta_{\sigma,\sigma}). \ee
\end{definition}
The functions  $k_\alpha(w) = \frac{1}{2}\left( w^{-\alpha} + w^{\alpha -1}\right)$, with $k_\alpha \in {\cal K}$,  
yield the  family of $\chi^2_\alpha$-divergences  given in \eqref{alpha-fam} which we call the
mean $ \alpha$-divergences  to distinguish them from the well-known family of WYD $ \alpha$-divergences, 
which are described in Appendix \ref{appA}, along with several other families. Although we focus on the family \eqref{alpha-fam},  
most of our results hold for  any divergence given by \eqref{chi2def} with $k \in {\cal K}$  with the exceptions of Theorem ~\ref{contrlb}.\\

\subsection{Monotone Riemannian metrics and generalized relative entropies.}

This definition of the $\chi^2$-divergence stems from the analysis of monotone Riemannian metrics.  
By Riemannian metric, we mean a positive definite bilinear form $M_\sigma(A,B)$ on the hermitian tangent hyperplane 
$\mathcal{TP}=\{A \in \M_d: A=A^\dag,\tr[A]=0\}$. The metric is monotone if for all quantum channels 
$T :\M_d \mapsto \M_d$, states $\sigma \in \sym^+_d$ and $A \in \mathcal{TP}$, $M_{T(\sigma)}(T(A),T(A))\leq M_\sigma(A,A)$. 
Petz showed showed that there is a one-to-one correspondence between the above metrics and a special class 
of convex operator functions, which correspond to $1/k$ in our notation. He furthermore was able to relate several generalized 
relative  entropies (which he defined much earlier \cite{Petz0} and referred to as  quasi-entropies)
to monotone Riemannian metrics \cite{Petz2,Petz3,PetzNew}. The reverse implication, 
that every monotone Riemannian metric stems from a generalized relative entropy was first proved 
by Lesniewski and Ruskai \cite{LR}. Taking advantage of the well-known integral representations of operator monotone and convex functions \cite{Bhatia} 
one can express the $\chi^2$-divergences as well as the relative entropies explicitly. 
We shall briefly repeat the key points of the analysis that are necessary for our understanding of the mixing-time and contraction 
analysis for cpt-maps.  \\
We need to consider the class of functions $\mathcal{G}$ by which we denote the set of continuous operator
 convex functions  from $\mathbb{R}^+$ to $\mathbb{R}$ that satisfy $g(1)=0$. Note that these functions can all be classified in
terms of the integral representation:
\be
	g(w) = a(w-1) + b(w-1)^2 + c\frac{(w-1)^2}{w} + \int_0^{\infty} \frac{(w-1)^2}{w+s}d\nu(s),
\ee 
where $a,b,c >0$ and the integral of the positive measure $d\nu(s)$ on $(0,\infty)$ is bounded. 
The generalized relative entropy for states $\rho,\sigma \in S^+_d$ was first defined in \cite{Petz4,Petz5}. 
\begin{definition}
Let $g \in \mathcal{G}$. The generalized quantum relative entropy is given by 
\be H_g(\rho,\sigma) = \tr{[\rho^{1/2}g(\Delta_{\sigma,\rho})(\rho^{1/2})]} \ee
when $supp(\rho) \subseteq supp(\sigma)$, and infinity otherwise, and where $\Delta_{\rho,\sigma}$ is again the modular operator. 
\end{definition}
We now recall without proof a theorem \cite{Petz2,Petz3,LR} relating the relative entropy and the monotone Riemannian metric, mirroring the classical result (\ref{chientcl}):
\begin{theorem}
For every $k \in \mathcal{K}$, there is a $g \in \mathcal{G}$ such that for a given $\sigma \in \sym_d$, and $A,B$ hermitian traceless, 
we get:
\bea 
M^k_\sigma(A,B) &=& \left . -\frac{\partial^2}{\partial \alpha \partial \beta}H_g(\sigma+\alpha A, \sigma+\beta B) \right |_{\alpha = \beta = 0} \\\nonumber
			&=& \left \langle \;A \;,\Omega_\sigma^k(B) \; \right \rangle.
\eea
and, $k$ is related to $g$ by 
\be k(w)=\frac{g(w)+wg(w^{-1})}{(w-1)^2} \ee
\end{theorem}
From this theorem follows a convenient integral representation of the inversion $\Omega_\sigma^k$, which is equivalent to (\ref{genInv}) \cite{LR}. 
\bea \label{intOmega}
	\Omega_\sigma^k = \int_0^{\infty} \left(\frac{1}{sR_\sigma + L_\sigma} + \frac{1}{R_\sigma + s L_\sigma}\right) N_g(s) ds, 
\eea
where $N_g$ denotes the singular measure $N_g(s) ds = (b_g + c_g)\delta(s)ds + d\nu_g(s)$. Note, that the relationship 
between k and g is not one-to-one. Indeed, by setting $\hat{g}(w)=wg(w^{-1})$, we get back the above relation. 
However, there is a one-to-one correspondence between each $k$ and a symmetric $g_s(w)=g(w)+wg(w^{-1})$, 
and hence between each metric and a symmetric relative entropy. \\
Note that the $\alpha$-subfamily of  (\ref{alpha-fam}) has the 
associated symmetric relative entropy: 
$g^{sym}(x)=\frac{(1-w)^2}{2}\left( w^{\alpha -1} + w^{-\alpha}\right)$, 
so that
\be 
H^{sym}_\alpha(\rho,\sigma)=\frac{1}{2}(H_\alpha(\rho,\sigma)+H_\alpha(\sigma,\rho))
\ee
where,
\be 
H_\alpha(\rho,\sigma)=\tr{[\rho^{2-\alpha}\sigma^{\alpha-1}+\rho^{1+\alpha}\sigma^{-\alpha}-2\rho^\alpha\sigma^{1-\alpha}]}. \notag
\ee
The integral representation (\ref{intOmega}) of the inversion $\Omega_\sigma^k$ allows for a partial ordering of different monotone Riemannian
metrics that follows from the set of inequalities:
\be 
\frac{2}{x+1} \leq \frac{1+s}{2}(\frac{1}{s+x}+\frac{1}{sx+1}) \leq \frac{x+1}{2x}.
\ee
for $s \in [0,1]$, and $x \in \R^+$. 
We therefore see that there exists a partial ordering for the inversions, with a lowest and highest element in the hierarchy. The lowest element
gives rise to the so called Bures metric. Thus,
\be
\Omega^{Bures}_\sigma = 2({R}_\sigma +{L}_\sigma)^{-1} \leq {\Omega}^k_\sigma \leq ({L}^{-1}_\sigma +{R}^{-1}_\sigma)/2 = \Omega^{\alpha=0}_\sigma 
\ee  
The $\chi^2$-divergence is recovered from the metric upon setting
$\chi^2_k(\rho,\sigma) \equiv M_\sigma^k(\rho-\sigma,\rho-\sigma)$. We are therefore left with a partial order for all possible
$\chi^2$-divergences with a smallest and largest element according to,
\be\label{hierarchi} 
	\chi^2_{Bures}(\rho,\sigma) \leq  \chi^2_{k}(\rho,\sigma) \leq \chi^2_{\alpha=0}(\rho,\sigma).
\ee
The defining attribute of the above set of metrics is their monotonicity under the action of quantum channels.
This was first shown by Petz \cite{Petz3}, and later a proof based on the integral representation of $\Omega_\sigma^k$ 
(\ref{intOmega}), and on Schwarz-type inequalities, was provided by Ruskai and Lesniewski in \cite{LR}. 
Due to its importance for the mixing time analysis we shall repeat it here. 
\begin{theorem}
\label{Theo:chi2-mono} 
For all $\sigma \in \sym_d$, $M^k_\sigma$ is monotone under the action of a quantum channel 
$T:\M_d\rightarrow\M_{d}$  for all $k \in \K$ and $A \in \M_d$,  i.e.
\be 
M^k_\sigma(A,A) \geq M^k_{T(\sigma)} \left(T(A),T(A)\right)
\ee 
\end{theorem}
\proof{ The monotonicity follows immediately from the integral representation of the inversion  $\Omega_\sigma^k$ in 
 (\ref{intOmega}), and an argument proved in the appendix,  Theorem \ref{proofMono}. The theorem states that for every 
 channel $T$ and for arbitrary $A$, we have 
 \be
 	\tr\left[A^\dagger \frac{1}{R_{\sigma} + sL_{\sigma}} A \right] = \tr\left[T\left(A^\dagger \frac{1}{R_{\sigma} + sL_{\sigma}} A\right) \right]  \geq
	\tr\left[T(A)^\dagger \frac{1}{R_{T(\sigma)} + sL_{T(\sigma)}} T(A) \right].
 \ee
\qed}

\subsection{Properties of the quantum $\chi^2$-divergence}
The fact that the quantum $\chi^2_k$-divergence can be used to bound the mixing time lies in the following Lemma, 
that upper bounds the trace distance which is the relevant distance measure in the mixing time definition.
\begin{lemma} \label{trbound}
For every pair of density operators $\rho, \sigma \in S_d$, we have that
\be ||\rho-\sigma||_1^2\leq\chi^2_k(\rho,\sigma)\ee
\end{lemma}
\proof{
If the support of $\rho$ is not contained in the support of $\sigma$, then the right hand side is $\infty$. We can therefore assume w.l.o.g. that $\sigma >0$ by restricting the analysis to the support space of $\sigma$. The trace norm $\| A \|_{1}$ of some matrix $A \in \M_d$ can be expressed as \cite{HornJohnson}$\|A\|_{1} = \max_{U \in U(d)} \tr[U A]$, 
where the maximum is taken over all unitaries acting  on the $d$-dimensional Hilbert space.  
Thus, for any inversion $\Omega^k_\sigma$:
\bea
	\|A\|_{1}^2 &=&    \max_{U \in U(d)}\tr[U A]^2  = \max_{U \in U(d)}\tr\left[U [\Omega^k_\sigma]^{-1/2}\circ [\Omega^k_\sigma]^{1/2} (A)\right]^2 \notag \\
			 &=&   \max_{U \in U(d)}\tr\left[[\Omega^k_\sigma]^{-1/2}(U) [\Omega^k_\sigma]^{1/2} (A)\right]^2 \\
			&\leq&  \tr\left[A^{\dagger} \Omega_\sigma^k(A)\right]\max_{U \in U(d)}\tr\left[U^{\dagger}[\Omega^k_\sigma]^{-1}(U)\right] \notag
\eea
Let us consider the Bures inversion given by $\Omega_\sigma^{Bures} = 2\left[L_\sigma + R_\sigma \right]^{-1}$. 
Clearly, its inverse is  $\left[ \Omega_\sigma^{Bures} \right]^{-1} = \frac{1}{2}\left[L_\sigma + R_\sigma \right]$. Therefore, for any unitary U, 
\bea
\tr\left[U^{\dagger}[\Omega^{Bures}_\sigma]^{-1}(U)\right] = \frac{1}{2} \left( \tr[U^\dagger \sigma U] + \tr[U^\dagger U \sigma] \right) = 1.
\eea
Setting $A = \rho -\sigma$ and observing that $\chi^2_{Bures} \leq \chi^2_k$ for all $k \in \K$ completes the proof.
\qed}

\begin{comment}
We are also able to bound the relative entropy in terms of the $\chi^2$-divergence.

\begin{proposition} \label{relentbd}
For every pair of density operators $\rho, \sigma \in S_d$, we have that
\be S(\rho,\sigma) \leq \chi^2_k(\rho,\sigma),\ee
where $S(\rho,\sigma)=\tr{[\rho(\log{\rho}-\log{\sigma})]}$ is the usual relative entropy.
\end{proposition}
%
\proof{
In \cite{RuskaiStillinger}, it was shown that for any $\gamma\in(0,1]$,
\be S(\rho,\sigma) \leq \frac{1}{\gamma}(\tr{\rho^{1+\gamma}\sigma^{-\gamma}}-1)\ee

Now, consider the basis in which $\sigma$ is diagonal, and write
$\sigma = \sum_k \mu_k |k \rangle \langle k|$ and $\rho = \sum_{kl} r_{kl} |k \rangle \langle l|$. Then, 
\bea \chi^2_k(\rho,\sigma) - S(\rho,\sigma) &\geq& \chi^2_{Bures}(\rho,\sigma) - 2(\tr{\rho^{3/2}\sigma^{-1/2}}-1)\notag \\ 
&=& 2\sum_{kl} \frac{r_{kl}r_{lk}}{\mu_k+\mu_l} - r_{kk} - 2 (\frac{r^{3/2}_{kk}}{\sqrt{\mu_k}}-r_{kk})\\
&\geq& \sum_k (\frac{r_{kk}}{\sqrt{\mu_k}} -\sqrt{r_{kk}})^2 \geq 0\notag \eea
where the penultimate inequality is obtained by rearranging terms and by dropping the positive terms $\sum_{kl} r_{kl}r_{lk}/(\mu_l+\mu_k)$ with $k \neq l$. 
\qed}
\end{comment}

We have already stated that the family of $\chi^2_\alpha$-divergences defined in (\ref{alpha-fam}) can be cast into the general framework of monotone Riemannian metrics. Because of its computational simplicity, and its special symmetry when $\alpha=1/2$, we consider its properties more specifically. It is possible for instance 
to show monotonicity of this subfamily using arguments from joint convexity. As the proof
is interesting in its own right, we give it here:
\begin{proposition}
$\chi^2_\alpha$ is jointly convex in its arguments for $\alpha\in[0,1]$. 
Moreover, it is monotone w.r.t. completely positive trace-preserving maps, i.e.,
\be \chi^2_\alpha(\rho,\sigma)\geq \chi^2_\alpha\big(T(\rho),T(\sigma)\big),\ee for every
quantum channel $T:\M_d\rightarrow\M_{d}$.
\end{proposition}
\proof{
A direct application of Corrolary 2.1 in \cite{Liebtensor} guarantees that $\chi_\alpha^2(\rho,\sigma)$ is jointly convex in its arguments for any $\alpha \in [0,1]$.  This in turn implies monotonicity w.r.t. cp-maps by a standard argument: let us represent $T$ as $T(\rho)={\rm tr}_E\big[U(\rho\otimes\psi)U^\dagger\big]$ where $\psi$ is a pure state (i.e. rank-one projection), $U$ a unitary and ${\rm tr}_E$ the partial trace over an `environmental' system of dimension $m$. If we take a unitary operator basis $\{V_i\}_{i=1,..,m^2}$ in $\M_m$ (orthonormal w.r.t. the Hilbert-Schmidt inner product), we can write
\be T(\rho)\otimes\1_m/m=\frac1{m^2}\sum_{i=1}^{m^2} (\1\otimes V_i)U(\rho\otimes\psi)U^\dagger(\1\otimes V_i^\dagger).\ee
However, since $\chi^2_\alpha\big(T(\rho),T(\sigma)\big)=\chi^2_\alpha\big(T(\rho)\otimes\tau,T(\sigma)\otimes\tau\big)$ in particular for $\tau=\1_m/m$, we can now apply joint convexity. With the help of the fact that for any unitary $W$ it holds that $(W\cdot W^\dagger)^\alpha=W(\cdot)^\alpha W^\dagger$ we obtain the claimed result.
\qed}  
Furthermore, we note that this subfamily also has a natural ordering. 
\begin{proposition}
For every $\rho,\sigma \in \sym_d$, $\chi^2_\alpha$ is convex in $\alpha$, and reaches a minimum for $\alpha=1/2$.
\end{proposition}
\proof{
First note that $\chi^2_{\alpha=0}(\rho,\sigma)=\chi^2_{\alpha=1}(\rho,\sigma)$. That the minimum is reached for $\alpha=1/2$ follows directly 
from the Cauchy-Schwarz inequality. Applied to our problem we get
\bea
   \tr{\left[\rho \sigma^{-1/2} \rho \sigma^{-1/2}\right]}^2 &=& \tr{\left[\rho \sigma^{(\alpha -1)/2} \sigma^{ -\alpha/2} \rho \sigma^{(\alpha -1)/2} \sigma^{ -\alpha/2} \right]}^2 
   \\\nonumber &\leq& \tr{\left[\rho\sigma^{-\alpha}\rho\sigma^{\alpha-1}\right]}^2
\eea

%$\tr{[A^{\dagger}]}^2 \leq \tr{[A^\dagger A]}\tr{[A^\dagger A]}$ which holds for all operators $A,B$. Applied to our problem  we get
%\bea \tr{\rho\sigma^{-\alpha}\rho\sigma^{\alpha-1}} &=& \tr{\big(\sigma^{\alpha-1/2}\big)^2\big(\sigma^{-\alpha/2}\rho\sigma^{-\alpha/2}\big)^2}\notag\\
%&\geq & \tr{\rho\sigma^{-1/2}\rho\sigma^{-1/2}}.\eea

To see convexity, consider the second partial derivative of $\chi^2_\alpha$ with respect to $\alpha$:
\bea \frac{\partial^2}{\partial \alpha^2}\chi^2_\alpha(\rho,\sigma)&=&\tr{\sigma^{\alpha-1}\rho \sigma^{-\alpha}(\rho \log^2{\sigma}+\log^2{\sigma}\rho-2\log{\sigma}\rho\log{\sigma})}\notag\\
&=& \sum_{kl} \mu^{\alpha-1}_k\mu^{-\alpha}_l(\log{\mu_k}-\log{\mu_l})^2|\langle k|\rho|l \rangle|^2 \geq 0\eea
where we used $\sigma=\sum_k \mu_k \ket{k}\bra{k}$.
 \qed}

Finally, we point out a bound on the relative entropy in terms of the $\alpha$-subfamily of $\chi^2$-divergences:

\begin{theorem} For every pair of density operators $\rho$ and $\sigma$ and every $\alpha\in(0,1]$ we have that
\be  \chi^2_\alpha(\rho,\sigma) \geq S(\rho,\sigma),\ee
where $S(\rho,\sigma) = \tr{\rho (\log{\rho}-\log{\sigma})}$ is the usual relative entropy.
\end{theorem}
\proof{
It was shown in \cite{RuskaiStillinger} that for $\gamma\in(0,1]$, the following holds:
\be S(\rho,\sigma) \leq \frac{1}{\gamma}(\tr{\rho^{1+\gamma}\sigma^{-\gamma}}-1)\ee \label{RSgamma}
Then consider,
\bea \chi_{\alpha}^{2}(\rho,\sigma)-S(\rho,\sigma) &\geq& \tr{\rho\sigma^{-1/2}\rho\sigma^{-1/2}}-2\tr{\rho^{3/2}\sigma^{-1/2}}+1 \notag\\
&=& \tr{(\rho^{1/2}\sigma^{-1/2}\rho^{1/2}-\rho^{1/2})^2} \geq 0\eea
where the first inequality comes from taking $\gamma=1/2$ in (\ref{RSgamma}), and $\alpha=1/2$ for $\chi_{\alpha}^2$, and the last line is obtained from rearranging terms.
\qed}

%
% Mixing time bounds and Contraction under CPT maps
%
\section{Mixing time bounds and Contraction of the $\chi^2$-divergence under CPT maps}
\label{contraction}

\subsection{Mixing time Bounds}
The $\chi^2$-divergence is an essential tool in the study of Markov chain mixing times, because on the one hand it bounds the trace distance, and on the other it allows easy access to the spectral properties of the map. The subsequent analysis can be seen as a generalization of the work presented in \cite{Diaconis,Fill} 
to the non-commutative setting. 

\begin{theorem}
\label{MixBounds}
 Let $T: M_d \mapsto M_d$ be an ergodic quantum channel with fixed point  $\sigma \in \sym_d$, 
 for any $\rho \in \sym_d$ and any $k \in \mathcal{K}$, we can bound
\be
 	\|T^{n}(\rho) - \sigma \|_{\tr} \leq (s^k_1)^{n}\sqrt{\chi_k^2(\rho,\sigma)}.
\ee
 Here $s^k_1$ denotes the second largest singular value (the largest being $1$) of the map 
\be Q_k =  [\Omega_\sigma^k]^{1/2} \circ T \circ [\Omega_\sigma^k]^{-1/2}\ee 
\end{theorem}
Before we prove Theorem \ref{MixBounds}, we would like to point out an important fact that regards the the singular 
values of $Q_k$. The monotonicity of the $\chi^2$-divergence ensures, that the singular values $s_i^k$ 
of $Q_k$ are always contained in $[0,1]$ irrespectively of the choice of $k \in \K$. Let us therefore prove the following:
\begin{lemma}\label{specInterval}
The spectrum of the map $S_k \equiv Q^*_k \circ Q_k = [\Omega_\sigma^k]^{-1/2} \circ T^* \circ \Omega_\sigma^k \circ T  \circ [\Omega_\sigma^k]^{-1/2}$ 
is contained in $[0, 1]$.  
\end{lemma}
\proof{
Let us first note, that the map $S_k$ is hermitian and positive by construction. Furthermore, the monotonicity of the 
$\chi^2$-divergence, as stated in Theorem \ref{Theo:chi2-mono} ensures that the Rayleigh-Ritz quotient is bounded by $1$.
This holds, since $\forall B$
\bea
&&\avr{B ,S_k (B)} = \avr{A, T^* \circ \Omega_\sigma^k \circ T(A)} =
 M^k_{T(\sigma)}(T(A),T(A)) \leq \\\nonumber && M^k_\sigma(A,A) =  \avr{A ,\Omega_\sigma^k(A)}
 = \avr{B,B},
\eea 
where we defined the intermediate state $A = [\Omega_\sigma^k]^{-1/2}(B)$.  
Note that we made use of the fact that $\sigma  = T(\sigma)$ is the
fixed point of the map. Therefore 
\bea
\lambda_{max} = \max_{B \in \M_d}\frac{\avr{B ,S_k (B)}}{\avr{B,B}} \leq 1
\eea
and the maximum is attained for $\lambda_{max} = 1$ and $B_{max} =  [\Omega_\sigma^k]^{1/2}(\sigma)$.
\qed}
\vspace{0.5cm}
With the bound on the spectrum at hand, it is now straight forward to prove Theorem \ref{MixBounds}
\vspace{0.5cm}

 \proof{
 Define $e(n) \in M_d$, as $e(n) = T^n(\rho -\sigma)$. By Lemma \ref{trbound}, we get $\|e(n)\|_1^2 \leq \chi^2_{k}(T^n(\rho),T^n(\sigma)) \equiv \chi^2_k(n)$.
 In the matrix representation, $\ket{e(n)} = e(n) \otimes \1 \ket{I}$, we can rewrite $\chi^2_k(n) = \bra{e(n)}\; \hat{\Omega}_\sigma^k \; \ket{e(n)}$.
 Note that also, $\ket{e(n+1)} = \hat{T} \ket{e(n)}$ and so,
\bea
\chi^2_k(n) - \chi^2_k(n+1) &=& \bra{e(n)} \; \hat{\Omega}_\sigma^k \;  \ket{e(n)} -
\bra{e(n)} \hat{T}^\dagger \; \hat{\Omega}_\sigma^k \; \hat{T} \ket{e(n)}\\
&=& \bra{e(n)}\; [\hat{\Omega}_\sigma^k]^{1/2} \; \left( \1 - \hat{Q_k}^\dagger \hat{Q_k}\right)
 [\hat{\Omega}_\sigma^k]^{1/2} \ket{e(n)}.
\eea
Due to Lemma \ref{specInterval} we know that the spectrum of $\hat{S}_k = \hat{Q_k}^\dagger \hat{Q_k}$,  
which is equal to the square of the singular values of $\hat{Q_k}$,  is contained in the interval $[0,1]$. Hence,
\bea
\bra{e(n)} [\hat{\Omega}_\sigma^k]^{1/2} \left(\1 -\hat{S_k}\right) [\hat{\Omega}_\sigma^k]^{1/2} \ket{e(n)} \\
\geq (1-s_1^2) \bra{e(n)} [\hat{\Omega}_\sigma^k]^{1/2} \sum_{\alpha \neq 0} P_{\alpha} \; [\hat{\Omega}_\sigma^k]^{1/2} \ket{e(n)}.
\eea
The sum is taken over spectral projectors $P^k_{\alpha}$ of $ \hat{S}_k = \sum_{\alpha}(s^k_\alpha)^2 P_{\alpha}$,
apart from $P^k_0$ which projects onto $ [\hat{\Omega}_\sigma^k]^{-1/2} \ket{\sigma}$. In particular, $P_0^k=[\hat{\Omega}_\sigma^k]^{-1/2} \ket{\sigma} \bra{I}[\hat{\Omega}_\sigma^k]^{-1/2}$, so that $\bra{e(n)} [\hat{\Omega}_\sigma^k]^{1/2} P^k_0 \; [\hat{\Omega}_\sigma^k]^{1/2} \ket{e(n)} = \bra{e(n)}\sigma \rangle \tr{[T^n(\rho-\sigma)]}=0$, by trace preservation of $T$. We can write,
\be
\chi^2_k(n) - \chi^2_k(n+1) \geq (1-(s^k_1)^2)\chi^2_k(n).
\ee
Rearranging terms completes the theorem. \qed}
\vspace{0.5cm}
{\bf Remark:}
The fact, that the singular values of $Q_k$ are always smaller or equal to one justifies the use of the generalized $\chi^2$-divergence
as the appropriate distance measure to bound the convergence of an arbitrary channel. It is tempting to use the 
Hilbert-Schmidt inner product to give an upper bound to the trace norm. This can always be done at the cost of a dimension dependent prefactor,
since on finite dimensional spaces all norms are equivalent. However, when doing so a problem  arises if one tries to bound the convergence in terms 
of the spectral properties of the map $ S_{HS} = T^* \circ T$. It is in general not ensured that the spectrum will be bounded by one.  
In fact, for every non-unital channel $T$, $S_{HS}$ will have an eigenvalue larger than one \cite{PerezWolf}. The similarity transformation of the channel $T$
with $[\Omega^k_\sigma]^{1/2}$ alters the singular values, but of course leaves the spectrum invariant. Furthermore, it is a well known fact \cite{HornJohnson} that the singular values of a square matrix log-majorize the absolute value of the eigenvalues. As the spectrum of $Q_k$ is bounded by one (and equal that of $\hat{T}$ by similarity), we conclude that its second largest eigenvalue is always smaller or equal to its second largest singular value. One can also give a general bound in 
terms of the second largest eigenvalue of $T$ \cite{Teheral}, but one is then confronted with a potentially severe dimensional prefactor.
\\\\
For some instances of the inversion $\Omega^k_\sigma$ it becomes immediately evident  that the symmetrization $S_k$ 
has the desired spectral properties without making use of the monotonicity of the $\chi^2_k$-divergence. 
It can occur, that $S_k$ is again similar to a quantum channel that is of the form 
$T^k_s =  [\Omega^k_\sigma]^{-1/2} \circ S_k \circ [\Omega^k_\sigma]^{1/2}$.  
A possible example of such an inversions is $ \Omega_\sigma^{\alpha = 1/2} = L_\sigma^{-1/2}R_\sigma^{-1/2}$.
This is however not the generic case, most inversions will lead to maps that are not completely positive any longer.
It would be very desirable to find other such examples, as they mirror the classical situation where the 
symmetrized maps are always probability transition matrices, and because these specific inversions allow 
for clean contraction bounds as seen in section III.B.\\\\

It is clear from the discussion above that the singular values of ${Q}_k$  play a crucial role in the mixing time analysis presented here. 
This seems to contradict the general understanding that the convergence is determined by the spectral  properties of the channel $T$ 
in the asymptotic limit. This can however be understood as follows: the matrix $\hat{Q}_k$  is  similar to $\hat{T}$, i.e. 
$\hat{Q}_k = [\hat{\Omega}^k_\sigma]^{1/2} \cdot \hat{T} \cdot [\hat{\Omega}^k_\sigma]^{-1/2}$, so the spectra of $Q_k$ and $T$ coincide.
The following lemma establishes a relation between the singular values and the eigenvalues in the asymptotic limit. 
For a proof, see e.g.  \cite{HornJohnson2} pg.180.
\begin{lemma}
Let $\hat{Q}_k \in \M_{d^2}$ be given, and let ${s_0(\hat{Q}_k) \geq \ldots \geq s_{d^2-1}(\hat{Q}_k)}$ and $\{\lambda_i(\hat{Q}_k)\}_{i=0\ldots d^2-1}$
denote its singular values and eigenvalues, respectively with $|\lambda_0(\hat{Q}_k)| \geq \ldots \geq |\lambda_{d^2-1}(\hat{Q}_k) |$. Then
\be
\lim_{n \rightarrow \infty} [s_i(\hat{Q}_k^n)]^{1/n} = |\lambda_i(\hat{Q}_k)| \;\;\; \forall \; i=0 \dots d^2-1
\ee
\end{lemma}
In the limit of $n \ra \infty$ applications of the quantum channel, we can start blocking the channel in $m$ subsequent applications $T^{(m)} \equiv T^m$ 
and bound the convergence rate as a function of the singular values of the corresponding $\hat{Q}^{(m)}_k$, which indeed converge to the 
eigenvalues of the original cp-map . Convergence following the eigenvalue is therefore only guaranteed in the limit of $n \ra \infty$, and this 
would indeed be the case, when e.g. the eigenstructure of the original cp-maps contains a Jordan block associated to the second largest eigenvalue. 
Note, that convergence in the above lemma goes typically as $1/n$, which is very slow. Hence for finite $n$, convergence is governed by the singular 
values of $\hat{Q}_k$ as opposed to the eigenvalues. The bound derived in (\ref{MixBounds}) is an absolute bound for finite $n$ and clearly leads 
to a strictly monotonic decay. Note that in the case that the second largest singular value is also equal to 1, this can then always be cured by 
blocking the cp-maps together. Finally, it is worth mentioning that the convergence can in fact be much more rapid if one starts in a state "closer" to the fixed point. In particular, if the initial state is such that $\rho-\sigma \propto Y_k$, $k \geq 2$, where $Y_k$ is the eigenvector corresponding to $\lambda_k$, then the convergence will be governed by the magnitude of $\lambda_k$. Furthermore, if instead of a single fixed point, we have a fixed subspace, or a collection of fixed subspaces (with or without rotating points), then the convergence to this fixed subspace will be governed by the largest eigenvalue whose magnitude is strictly smaller than one.\\

Thus far we have only considered the time-discrete case, it is however straightforward to give a similar bound for time-continuous
Markov processes, that are described by a one parameter semi-group. The following lemma bounds the trace-distance as a function of $t \in \R^+_0$:
The proof of the following lemma is very similar to the proof of the time discrete case, we will therefore omit it here.
\begin{lemma}
Let $\lio$ denote the generator of a time continuous Markov process, described by the master equation $\partial_t \rho = \lio(\rho)$,
with solution $\rho(t) \in \sym_d$ $\forall \; t \in [0,+\infty)$ . Furthermore let $\sigma \in \sym^+_d$ denote the fixed-point $\lio(\sigma) = 0$, then
\be 
\left \| \rho(t) - \sigma \right \|_{\tr}^2 \leq e^{l^k_1 t } \chi_k^2(\rho(0),\sigma).
 \ee
Here, $l^k_1 \leq 0$ refers to the second largest eigenvalue of
\be \Lambda_k =    [\Omega_\sigma^k]^{1/2} \circ \lio^* \circ [\Omega_\sigma^k]^{-1/2}  + 
[\Omega_\sigma^k]^{-1/2} \circ \lio \circ  [\Omega_\sigma^k]^{1/2}.
\ee
\end{lemma}

The symmetrization for the generator of  the time continuous Markov process is additive as would be expected. Furthermore, we note that the monotonicity of the 
$\chi^2$-divergence ensures that the spectrum of $\Lambda_k$ is never positive, based on a similar reasoning as given in lemma (\ref{specInterval}). 

\subsection{Contraction Coefficients}
In the following we study the contraction of the $\chi^2$-divergences under quantum channels, and its relation to the trace norm contraction. We consider general contraction rather than contraction to the fixed point because analytic results are more readily available, and because these bounds are in a sense the \textit{most} stringent one can require. We focus primarily on the mean $\alpha$-subfamily of $\chi^2$-divergences. \\ Let us define the following contraction coefficients which we call the $\chi^2$- and trace norm- contraction respectively:
\be \eta^\alpha_{\chi}(T) = \sup_{\rho,\sigma \in \sym_d} \frac{\chi_{\alpha}^2(T(\rho),T(\sigma))}{\chi_{\alpha}^2(\rho,\sigma)} \ee
and 
\be \eta_{\tr}(T)=\sup_{\rho,\sigma \in \sym_d} \frac{||T(\rho-\sigma)||_1}{||\rho-\sigma||_1} = \sup_{\phi,\psi \in \sym^1_d ,\langle \phi|\psi \rangle=0} \frac{1}{2}||T(\psi)-T(\phi)||_1, \ee

where $T:\M_d\rightarrow\M_{d}$ is a quantum channel, and the last equality is seen simply by convexity of the trace norm.\\
We first upper bound the trace-norm contraction in terms of the $\chi^2$ contraction, which is a generalization of a result in \cite{Ruskai1}:
\begin{theorem}\label{trnormub}
For all $\alpha \in (0,1]$, and a quantum channel $T:\M_d\rightarrow\M_{d}$, 
\be \eta_{tr}(T) \leq \sqrt{\eta^\alpha_\chi(T)} \ee.
\end{theorem}
\proof{ From Lemma \ref{trbound}, we have that $||T(\rho-\sigma)||^2_1 \leq \chi^2_\alpha(T(\rho),T(\sigma))$, for all $\rho, \sigma \in \sym_d$. Let $N$ be traceless and hermitian, and note that it can be written as $N=N_{+}-N_{-}$, where $N_{+}, N_{-}$ are positive definite and orthogonal in their support. Now let $P=|N|/||N||_1$ and recall that $|N|=N_{+}+N_{-}$, then we get $\tr{[NP^{-\alpha}NP^{\alpha-1}]} = ||N||_1^2$, for every $\alpha \in (0,1]$. Also,
\be\frac{||T(N)||^2_1}{||N||^2_1} \leq \frac{\tr{[T(N)T(P)^{-\alpha}T(N)T(P)^{\alpha-1}]}}{\tr{[NP^{-\alpha}NP^{\alpha-1}]}} \ee
where the inequality is in the numerator, and the denominators are equal, by the previous observation. Taking the supremum over all traceless hermitian $N$ on the left hand side and identifying $\rho-\sigma=N/||N||_1$, $P=\sigma$ then gives the desired result. \qed}

We now provide a lower bound to the trace norm contraction for primitive channels:
\begin{theorem} \label{contrlb}
Given a quantum channel $T:\M_d\rightarrow\M_{d}$,
\be \eta^{\alpha=1/2}_\chi(T) \leq \eta_{tr}(T)\ee
\end{theorem}
First we introduce an eigenvalue type min-max characterization of the $\chi^2$-contraction, and then show that this eigenvalue must be smaller than the trace norm-contraction.\\
Let $P>0$, and consider the following eigenvalue equation:

\be \hat{\Gamma}\ket{A} \equiv \hat{\Omega}^{-1}_P\hat{T}^\dag \hat{\Omega}_{T(P)} \hat{T} \ket{A} = \lambda \ket{A}  \label{evalw}, \ee 

where $\Omega_X \equiv \Omega^{\alpha=1/2}_X$. It $T$ has a non-trivial kernel, then $\Omega_{T(P)}$ should be understood in terms of the pseudo-inverse. First note that $\Gamma$ is a quantum channel, so its spectrum is bounded by one, and that it reaches one for $A=P$. Also note that $\Gamma$ is similar to a hermitian operator, so it has all real eigenvalues, so we can take the eigenvectors to be hermitian. Then rewriting (\ref{evalw}) as $\hat{T}^\dag \hat{\Omega}_{T(P)} \hat{T} \ket{A} = \lambda \hat{\Omega}_P \ket{A}$, we can express the second largest eigenvalue as:

\bea \lambda_1(T,P) &=& \sup_{ \bra{N} \Omega_P(P) \rangle=0,N=N^\dag} \frac{\bra{N}\hat{T}^\dag \hat{\Omega}_{T(P)} \hat{T} \ket{N}}{\bra{N} \hat{\Omega}_P \ket{N}} \notag \\
&=& \sup_{\tr{N}=0,N=N^\dag} \frac{tr[ T(N)T(P)^{-1/2}T(N)T(P)^{-1/2}]}{tr[NP^{-1/2}NP^{-1/2}]}. \eea

Clearly, by maximizing over all $P$, one recovers $\eta^{1/2}_\chi(T)$. We now prove the above theorem:\\
\proof{
Let $N_1$ be the eigenvector for which $\lambda_1$ satisfies the eigenvalue equation (\ref{evalw}), and recall that $N_1$ is Hermitian and traceless. 
Then,
\be \lambda_1 ||N_1||_1 = ||\Gamma(T(N_1))||_1 \leq ||T(N_1)||_1 \ee
because $\Gamma$ is a channel, and 
\be \lambda_1 \leq \frac{||T(N_1)||_1}{||N_1||_1} \leq \sup_{\tr{N}=0, N^{\dag}=N} \frac{||T(N)||_1}{||N||_1}=\eta_{tr}, \ee
taking the supremum over positive $P$ completes the proof. 
\qed}

\vspace{0.5cm}
{\bf Remark:}
Theorem~\ref{contrlb} gives a computable lower bound to the trace norm contraction. 
A key subtlety in the argument is that $\left[\Omega_P(A)\right]^{-1} = {\sqrt{P}} A {\sqrt{P}}$ 
is a completely positive (CP) map (with a single Kraus operator $\sqrt{P}$) which implies that 
$\Gamma$ is a quantum channel.  In general, $\Omega_P$ is not even positivity preserving.   
Another exception is the monotone metric associated with the usual logarithmic relative 
entropy for which $k(w) = \frac{ \log w}{w-1}$.
It is well-known \cite{AN, Petz0, LR} that  $  \Omega^{\log}_P(A)$ can be written as
\bea   \label{omegalog}
   \Omega^{\log}_P(A) = \int_0^\infty \frac{1}{P + xI} A   \frac{1}{P + xI}  dx
\eea
which is clearly CP. An analogous lower bound was shown in \cite{LR} for this map using a
similar argument.    Clearly, this can be extended to any monotone metric for which
$ \Omega_P$ is CP; however, we do not know of any other examples.

Very little is known about the ordering of the general $\eta_k$ contraction coefficients. 
In particular, We do not know whether  whether 
$\eta^{\log}_\chi$ is smaller or larger than $\eta^{\alpha=1/2}_\chi$. 
However, it is  known \cite{LR} that $\eta_k$ are not all identical for different $k \in \mathcal{K}$.;
because examples can be constructed using non-unital qubit channels.
Theorem~\ref{trnormub} can readily be extended to any metric associated
with  $k \in {\cal K}$.    However, it seems unlikely that Theorem~\ref{contrlb} holds in general,.
Thus, we can conclude
\bea
   \max \{  \eta_\chi^{\alpha=1/2}(T), \eta_\chi^{\log}(T)  \}  \leq   \eta_\tr(T)  \leq  \inf_{k \in {\cal K}}  \sqrt{\eta_\chi^k(T)} ~.
\eea
 \\

Note that if instead of maximizing over all $P$ we only consider contraction of the map to the steady state, and denote it $\bar{\eta}(T)=\eta(T)_{P=\sigma}$, then from the above arguments one immediately gets:
\be \bar{\eta}^\alpha_\chi(T) \leq \bar{\eta}_\tr(T) \leq \eta_\tr(T) \leq 1 \ee \label{etabar}
Combing this with the previous bounds above, we have
\be 
\lambda_1 \leq s_1^{\alpha=1/2} = \bar{\eta}_\chi^{\alpha=1/2} \leq \eta_\chi^{\alpha=1/2} \leq \eta_\tr \leq \sqrt{\eta_\chi^{\alpha=1/2}}. \ee
Moreover, $k(w) = \sqrt{w}$ on the right can be replaced by any $k \in {\cal K}$,
and that on the left  by $k(w)=(w-1)^{-1}\log{w}$.   It is very tempting to conjecture that $\bar{\eta}^2_\tr \leq \bar{\eta}^\alpha_\chi$, 
and/or that $\eta_\tr \leq \sqrt{\bar{\eta}_\chi^{\alpha=1/2}}$, but simple numerical counterexamples show these to be false.

% \pagebreak

\section{Quantum Detailed Balance}
\label{detailedbalance}

The detailed balance condition is often crutial in the analysis of classical Markov chain mixing times, as it ensures several convenient properties of the Markov chain. 
In particular, it implies that the classical probability distribution with respect to which the stochastic map 
is detailed balanced is a fixed point of the chain. Furthermore, detailed balanced stochastic 
maps have a real spectrum.  In this section we generalize the notion of classical detailed balance 
to quantum Markov chains. Alternative definitions of quantum detailed balance have been given in the literature: \cite{refdetBal, Majewski,Streater,Alicki} and references therein.
Central to our approach is the operator $Q_k$ as previously introduced 
in Lemma \ref{MixBounds}. In the literature for classical Markov chains an analogous matrix 
exists and is often referred to as the discriminant. 
\begin{definition}
For a channel $T:\M_d\rightarrow\M_{d}$ and a state $\sigma \in \sym^+_d$ with corresponding inversion 
$\Omega_{\sigma}^k$ as defined in (\ref{genInv}), we define the quantum discriminant of $T$ as,
\be \label{Discriminant}
 	Q_k =  [\Omega_\sigma^k]^{1/2}  \circ T \circ [\Omega_\sigma^k]^{-1/2}.
\ee
\end{definition}
We recall that the convergence of an arbitrary quantum Markov process can be bounded by the singular values of  $\hat{Q}_k$. 
Classical detailed balanced Markov chains have the property that the corresponding discriminant becomes symmetric. 
We shall therefore define the quantum generalization by requiring that for a quantum detailed balanced process 
\be
Q_k^* = Q_k.
\ee
This immediately allows to make a statement about the spectrum of quantum detailed balanced maps. 
Due to the hermicity of the matrix representation of the map (\ref{Discriminant}) we can immediately deduce, just as for classically case,
that the quantum channel $T$ has a real spectrum. For detailed balanced maps, the second largest eigenvalue in magnitude coincides with the second largest singular value. 
Furthermore, we would like to point out that this is actually not just a single condition for quantum detailed balance but a whole family. 
Hence every different inversion $\Omega^k_{\sigma}$ gives rise to a different condition for detailed balance. We therefore 
define as the quantum generalization of detailed balance:
\begin{definition}\label{def:DetBalance}
For a channel $T:\M_d\rightarrow\M_{d}$ and a state $\sigma \in \sym^+_d$, we say that $T$ obeys {\emph k- detailed balance} with respect to $\sigma$ with  $k \in \K$ , when
\be\label{eqDetBal}
	 [\Omega_{\sigma}^k]^{-1} \circ T^{*}  = T \circ [\Omega_{\sigma}^k]^{-1}.
\ee
\end{definition}
A consequence of this definition is that $\sigma$ is a fixed point of $T$.
\begin{lemma}\label{lem:steady}
Let $\sigma \in \sym_d$ be a state and $T$ a channel that
satisfies the detailed balance definition  \ref{def:DetBalance} with respect to $\Omega_{\sigma}^k$, 
then $\sigma$ is a steady state of $T$.
\end{lemma}
\begin{proof}
Recall that the inverse is given by $[\Omega_{\sigma}^k]^{-1} = R_{\sigma} f(\Delta_{\sigma,\sigma})$, where $f(w) = 1/k(w)$. 
Hence, since $k(1) = f(1) = 1$, we have
\bea
[\Omega_{\sigma}^k]^{-1}(\1) = R_{\sigma} f(\Delta_{\sigma,\sigma}) \1 = R_{\sigma} \1 =  \sigma.
\eea
Now, since furthermore $T^*(\1) = \1$, we have that 
\bea
	T(\sigma) =  T \circ \left[ \Omega_{\sigma}^k \right]^{-1}(\1) = 
	\left[ \Omega_{\sigma}^k \right]^{-1} \circ T^{*}(\1) = [\Omega_{\sigma}^k]^{-1}(\1) = \sigma.
\eea
\qed
\end{proof}
Given a probability distribution on some set of states, it is desirable to have a simple criterium to check whether
a completely positive map obeys detailed balance with respect to the state generated from the distribution. This
criterium may then serve to set up a Markov chain that actually converges to the desired steady state.
\begin{proposition}\label{DetBalcond}
Let $\{\ket{i}\}_i$ be a  complete orthonormal basis of $\HS$ and let  $\{\mu_i\}_i$ be a probability distribution on this basis.
Furthermore, assume that a quantum channel $T$ obeys
\be
\frac{\mu_n}{ k\left(\mu_m / \mu_n\right)} \;\;  \bra{i}\; T(\;\ket{n}\bra{m}\;) \; \ket{j}
 = \frac{\mu_i}{ k\left(\mu_j / \mu_i \right)}  \;\;   \bra{m}\; T(\;\ket{j}\bra{i}\;) \; \ket{n},
\ee
then $\sigma = \sum_i \mu_i \ket{i}\bra{i}$ and $T$ obey the detailed balance condition with respect to $\Omega_{\sigma}^k$.
\end{proposition}
\begin{proof}
Note that $\{\ket{i}\bra{j}\}_{ij}$ forms a complete and orthonormal basis in the space $\M_d$ with respect to the Hilbert-Schmidt scalar 
product. We can therefore express equation (\ref{eqDetBal}) in this basis. The individual entries are equal due to
\bea
\tr\left[ \left(\ket{m}\bra{n}\right)^{\dagger}  [\Omega^k_\sigma]^{-1} \circ T^*(\ket{j}\bra{i})\right] = 
\mu_n \; k^{-1}\left(\mu_m / \mu_n\right) \tr\left[ T(\; \ket{m}\bra{n}\; )^{\dagger} \left(\ket{j}\bra{i}\right)\right] = \\\nonumber
\mu_n\; k^{-1}\left(\mu_m / \mu_n\right) \;\;  \bra{i}\; T(\;\ket{n}\bra{m}\;) \; \ket{j} = 
\mu_i\; k^{-1}\left(\mu_j / \mu_i \right)  \;\;   \bra{m}\; T(\;\ket{j}\bra{i}\;) \; \ket{n}  = \\\nonumber
\mu_i\; k^{-1}\left(\mu_j / \mu_i \right) \tr\left[  \left(\ket{m}\bra{n}\right)^{\dagger}  T(\ket{j}\bra{i})\right] = 
\tr\left[ \left(\ket{m}\bra{n}\right)^{\dagger} T \circ [ \Omega^k_\sigma]^{-1}(\ket{j}\bra{i})\right].
\eea
\qed
\end{proof}
{ \bf Remark:}
We note that the different quantum detailed balance conditions coincide for classical channels,  i.e. for stochastic processes that 
are included in the framework of quantum channels. Define the following "classical" Kraus operators:
\bea
	A^{cl}_{ij} = \sqrt{P_{ij}} \ket{i}\bra{j}  \; \; \mbox{and a state,} \; \;  \sigma = \sum_i \mu_i \ket{i}\bra{i}.
\eea
In this case, the condition of Proposition \ref{DetBalcond} reduces to the classical condition. This can be seen when considering 
the channel $T^{cl}(\rho) = \sum_{ij} A^{cl}_{ij} \rho A^{cl\dagger}_{ij}$ and checking for detailed balance with respect to sigma, since 
\be
\frac{\mu_m}{k\left(\mu_n / \mu_m\right)} \;\;  \bra{i}\; T^{cl}(\;\ket{n}\bra{m}\;) \; \ket{j} = \frac{\mu_m}{ k\left(\mu_n / \mu_m\right)} \delta_{nm} \delta_{ij} P_{in} \\\nonumber
\ee
and
\be
\frac{\mu_i}{ k\left(\mu_j / \mu_i\right)} \;\;  \bra{i}\; T^{cl}(\;\ket{n}\bra{m}\;) \; \ket{j} = \frac{\mu_i}{ k\left(\mu_j / \mu_i\right)} \delta_{nm} \delta_{ij} P_{ni}.
\ee
However since $k(1) = 1$ we are just left with the classical detailed balance condition $\mu_i P_{ni} = \mu_n P_{in}$  for all pairs $i,n$. \\\\
A natural question to ask is therefore, whether the different detailed balance condition are all identical. To see 
that this is not the case, consider the example given by the Kraus operators of a single qubit, i.e. $\HS = \cc^2$,
\bea
	A_1 = \frac{1}{\sqrt{2}}\left( \begin{array}{cc} 1 & 1 \\ 0  & 0 \end{array} \right) \;\;\; \mbox{and} \;\;\; 
	A_2 = \frac{1}{2}\left( \begin{array}{cc} 1 &  -1 \\  1  & -1 \end{array} \right).
\eea
This channel has the unique fixed point 
\bea
	\sigma =  \frac{1}{6} \left( \begin{array}{cc}  5 &  1 \\  1  & 1 \end{array} \right).
\eea
From this channel it is now possible to construct a channel that obeys detailed balance with respect to the inversion given 
by choosing $k(w) = w^{-1/2}$, that is the inversion reads $\Omega_\sigma^{\alpha=1/2} = L_\sigma^{-1/2}R_\sigma^{-1/2}$.  We consider 
therefore the symmetrized map, 
\bea
	T_s = \left[\Omega_\sigma^{\alpha=1/2} \right]^{-1} \circ  T^*  \circ  \Omega_\sigma^{\alpha=1/2} \circ T.  
\eea
For the specific instance where $\Omega_\sigma^{\alpha=1/2}$ is given as above, we are assured that the map $T_s$ is again 
a quantum channel, because one immediately finds the Kraus representation for 
$T_s(\rho) = \sum_{ij} B_{ij} \rho  B_{ij}^\dagger$  as $B_{ij}  = \sqrt{\sigma} A_i^\dagger [ \sqrt{\sigma}]^{-1} A_j$.
The individual Kraus operators read,
\bea
	&B_{11} = \frac{3}{5}\left( \begin{array}{cc} 1 & 1 \\ 1/2  & 1/2 \end{array} \right) \;\;\; &\mbox{and} \;\;\; 
	B_{12} = \frac{\sqrt{2}}{5}\left( \begin{array}{cc} 1 &  -1 \\  1/2  & -1/2 \end{array} \right), \\\nonumber\\\nonumber
	&B_{21} = \frac{\sqrt{2}}{20}\left( \begin{array}{cc} 3 & 3 \\ -1  & -1 \end{array} \right) \;\;\; &\mbox{and} \;\;\; 
	B_{22} = \frac{1}{5}\left( \begin{array}{cc} 3 &  -3 \\  -1  & 1 \end{array} \right).
\eea
The channel $T_s$ satisfies detailed balance with respect to $\Omega_\sigma^{\alpha=1/2}$ by construction.  This channel
however does not satisfy detailed balance with respect to the inversion 
$\Omega^{Bures}_\sigma =  2 \left[L_\sigma + R_\sigma\right]^{-1}$ as can be seen directly by evaluating the detailed balance 
condition in terms of the matrix representations,
\bea
\left [ \hat{\Omega}_\sigma^{Bures} \right]^{-1} \cdot \hat{T}_s^\dagger -  \hat{T}_s \cdot  \left [ \hat{\Omega}_\sigma^{Bures} \right]^{-1}  
= \frac{7}{600}\left[  \1 \otimes Y + Y \otimes \1 \right],
\eea
where 
\bea
	Y = \left( \begin{array}{cc} 0  &  -1 \\  1  & 0 \end{array} \right).
\eea
The family of quantum detailed balance conditions is therefore much richer than the classical counterpart. 

%
% Section - Quantum Cheeger's Inequality
%
\section{Quantum Cheeger's Inequality}
\label{cheeger}

In the context of classical stochastic processes a very powerful formalism has been developed, often referred to as the conductance bound or 
Cheeger's inequality, to bound convergence rates of stochastic processes. We will generalize this to the quantum setting in this section. 
Similar results have appeared in \cite{Hastings}.  The gap of the map $S_k$ is defined as the difference between 
the largest and second largest eigenvalue, $\Delta = 1 - \lambda_1$.  The gap can be characterized in a variational fashion \cite{HornJohnson}.

\begin{proposition}\label{vargap} 
The gap of the map $S_k = [\Omega_\sigma^k]^{-1/2} \circ T^* \circ \Omega_\sigma^k \circ T  \circ [\Omega_\sigma^k]^{-1/2}$ 
is given by
\bea
\Delta = \min_{X \in M_d} \frac{\left\langle X, (\mbox{id}-S_k) X \right \rangle}
{\frac{1}{2}\left \|( X \otimes \sqrt{\sigma} - \sqrt{\sigma} \otimes X) \right\|_{HS}^2},
\eea
where $\left\|A \right \|_{HS}^2 = \tr[A^\dagger A]$ denotes the standard Hilbert-Schmidt norm and $\langle\; , \rangle$ the corresponding 
Hilbert-Schmidt scalar product.
\end{proposition}
\proof{
The eigenvector that corresponds to the eigenvalue $\lambda_0 = 1$ of $S_k$ is given by $\sqrt{\sigma}$. 
The gap can therefore be written as\cite{HornJohnson}:
\bea
\Delta &=& \min_{X \in M_d; \tr[X\sqrt{\sigma}] = 0} 1 - \frac{\tr[X^\dagger S(X)]}{\tr[X^\dagger X]} \notag\\
&=& \min_{X \in M_d; \tr[X\sqrt{\sigma}] = 0} \frac{\tr[X^\dagger(X - S(X))]}{\tr[X^\dagger X] - \tr[X\sqrt{\sigma}]^2}\\\nonumber
&=& \min_{X \in M_d} \frac{\tr[X^\dagger(X - S(X))]}{\frac{1}{2}\left\|(X \otimes\sqrt{\sigma}- \sqrt{\sigma}\otimes X )\right \|_{HS}^2},
\eea
Note that the constrained $\tr[X\sqrt{\sigma}] = 0$ can be dropped in the last line. Suppose that $\tr[X\sqrt{\sigma}] = c$, we can then
define $X' = X - c\sqrt{\sigma}$ and vary $X'$ since the equation is invariant under such shifts.\qed}

Throughout the remainder of this section we consider unital quantum channels,  i.e. maps which obey $T(\1) = \1$.  For this case it is ensured that already the simple   map $S = T^{*} \circ T$ has a spectrum that is contained in $[0,1]$, since all $\Omega_\sigma^k$ coincide and correspond to the identity map. 
The $\chi^2$-divergence just reduces to the standard Hilbert-Schmidt inner product times a prefactor given by $d = \dim(\HS)$. In the case of a detailed balanced stochastic map it even suffices to just consider the map itself. In either case we will denote the corresponding map as $S$ from now on. 
The variational characterization of the gap $\Delta$ now allows us to give an upper as well as a lower bound to the second largest eigenvalue of $S$. 
\begin{lemma}\label{Lemma:conductanceBound}
Let $T :\M_d \rightarrow \M_d$ be a unital quantum channel. Then the second largest
eigenvalue $\lambda_1$ of its symmetrization $S = T^{*} \circ T$,
is bounded by,
\be
1-2h \leq \lambda_1 \leq 1 - \frac{1}{2}h^2,
\ee
where $h$ is Cheeger's constant defined as,
\be\label{ChegC}
h = \min_{\Pi_A,\tr[\Pi_A] \leq d/2}\frac{\tr\left[ \left(\1 - \Pi_A\right)S(\Pi_A) \right]}{\tr\left[\Pi_A\right]}.
\ee
The minimum is to be taken over all projectors $\Pi_A$  on the space $A \subset \HS$, so that $\tr[\Pi_A] \leq d/2$.
\end{lemma}

\begin{proof}
An upper bound to the gap is immediately found by choosing $X = \Pi_A$. Due proposition (\ref{vargap}) we can write:
\bea \Delta &\leq& \frac{\tr{[\Pi_A (\mbox{id}-S)(\Pi_A)]}}{\tr{[\Pi^2_A ]}-\frac{1}{d}\tr{[\Pi_A]}^2}\notag\\
&=&\frac{\tr{[(\1-\Pi_A)S(\Pi_A)]}}{\frac{1}{d}\tr{[(\1-\Pi_A)]}\tr{[\Pi_A]}}\leq 2h,\eea
where in the last line we have used that $\tr{[\1-\Pi_A]}\geq d/2$.\\
For the lower bound, we can restrict the minimization in (\ref{ChegC}) to diagonal projections. 
Furthermore, when considering only unital quantum channels, it is possible to reduce
the problem of bounding the gap $\Delta$ to that of bounding the gap of a classical stochastic process. 
To see this, let us work in the basis where the eigenvector $X_1 \in \M_d$
corresponding to $\lambda_1$ is diagonal.  We shall assume wlog that $X_1^\dagger = X_1$. 
In this basis, we can write $X = \sum x_i \ket{i}\bra{i}$. The numerator then becomes
\bea
\tr\left[X^\dagger(X - S(X))\right] = \sum_{ij} x_i x_j(\tr[\ket{i}\bra{i}\ket{j}\bra{j}]- \tr\left[\ket{i}\bra{i}S(\ket{j}\bra{j})\right]\notag \\
= \sum_i x_i^2 - \sum_{ij} x_ix_j P_{ij} = \frac{1}{2}\sum_{ij}P_{ij}(x_i-x_j)^2.
\eea
We introduced the matrix $P_{ij} = \bra{i}S(\ket{j}\bra{j})\ket{i}$, which is a symmetric non-negative matrix which obeys $P_{ij} \geq 0$ ,
$\sum_i P_{ij} = 1$ and $P^T = P$. Hence $P$ is doubly stochastic. Performing the same reduction in the denominator we obtain
\bea
\frac{1}{2d}\left\|(X\otimes \1 - \1 \otimes X)\right\|_{HS}^2 = \frac{1}{2d}\;\sum_{ij}(x_i - x_j)^2
\eea
Hence, we arrive at the classical version of Mihail's Identity \cite{Mihail},
\be
\Delta = \min_{\{x_i\}}\frac{\sum_{ij}P_{ij}(x_i-x_j)^2}{1/d\;\sum_{ij}(x_i-x_j)^2}.
\ee
Given the classical version of Mihail's identity, the proof of the lower bound is the same as in the classical case. For completeness we repeat it here.
\begin{comment}The constraints given by $\tr[X] = 0$ can written as $\sum_i x_i = 0$. It is important to note, that this constraint can be dropped in this
variational characterization, because the cost function has become independent of the constraint. Suppose $\sum_i x_i = c$, then 
the transformation $x'_i  =  x_i - c/d$ leaves the cost function invariant and ensures that the  $x'_i$ fulfill the constraint. \\

Given the classical version of Mihail's identity, the proof is the same as in the classical case. For completeness we will repeat the proof here.
It is straight forward to find an upper bound to the gap $\Delta$ in terms of  Cheeger's constant. We define for some subset 
$A  \subset  \{1,\ldots,d\} = A \cup A^c$, so that $\sum_{i \in A}\tr[\Pi(i)] \leq d/2$ a $ Y \in M_d$ as  
$Y = \frac{1}{d}(\sum_{i \in A}\Pi(i) - \sum_{i \in A^c}\Pi(i))$ and write $Y = \sum_i y_i \Pi(i)$.
For this $Y$, we immediately find,
\bea
\Delta \leq \frac{\sum_{ij}P_{ij}(y_i-y_j)^2}{1/d\sum_{ij}(y_i-y_j)^2} = \frac{\sum_{i \in A}\sum_{j \in A^c} P_{ij}}{\frac{1}{d} \; \sum_{i \in A^c} 1 \sum_{j \in A} 1}
\leq 2 \frac{\sum_{i \in A}\sum_{j_ \in A^c}P_{ij}}{\sum_{i \in A} 1}.
\eea
Note, that when defining the projector $\Pi_A =  \sum_{i \in A} \Pi(i)$ we can write,
\be
\frac{\sum_{i \in A}\sum_{j_ \in A^c}P_{ij}}{\sum_{i \in A} 1} = \frac{\tr\left[(\1 - \Pi_A)S(\Pi_A)\right]}{ \tr\left[\Pi_A\right]},
\ee
for all $\Pi_A$, and hence $\Delta \leq 2h$.\\
\end{comment}
First, we define, $z_i \equiv |x_i| x_i$ and write,
\bea
\sum_{ij}P_{ij}|z_i - z_j| = \sum_{ij} P_{ij} ||x_i|x_i -  |x_j|x_j| \leq \sum_{ij} \sqrt{P_{ij}}\sqrt{P_{ij}} (|x_i| + |x_j|) (x_i - x_j)\notag\\
\leq \sqrt{\sum_{ij} P_{ij} (x_i - x_j)^2}\sqrt{\sum_{ij} P_{ij} (|x_i| + |x_j|)^2 },
\eea
where we used Cauchy-Schwartz in the last step. Consider now,
\be
\sum_{ij} P_{ij} (|x_i| + |x_j|)^2 = 2(\sum_i x_i^2 + \sum_{ij}|x_i|P_{ij} |x_j|) \leq 4 \sum_i |x_i|^2.
\ee
Furthermore, note that we can bound,
\be
1/d \sum_{ij}(x_i -x_j)^2 \leq 2/d \sum_{ij} x_i^2 = 2\sum_{i}|z_i|.
\ee
We are therefore left with a lower bound to Mihail's identity, which holds for all choices of $\{x_i\}$
\be
\frac{1}{2}\left(\frac{\sum_{ij}P_{ij}|z_i - z_j|}{2\sum_i |z_i|} \right)^2 \leq \frac{\sum_{ij}P_{ij}(x_i-x_j)^2}{1/d\sum_{ij}(x_i-x_j)^2}.
\ee
We shall now assume, that $x_i \geq 0$ everywhere and we can hence drop the absolute values in the definition for the $z_i$.
This is assumption is valid since we are free in adding an arbitrary constant $x_i \rightarrow x_i  + c$ to make all $x_i$ positive.  
Note that we therefore are left with a lower bound to the gap of the form, 
\be
\Delta \geq \frac{1}{2}\left(\frac{\sum_{ij}P_{ij}|x_i^2 - x_j^2|}{2\sum_i x_i^2} \right)^2
\ee
Let's focus on the right side of the inequality. Since,
\be
2 \sum_{i,j\;:\;x_i \geq x_j} P_{ij}(x_i^2 - x_j^2) = 4 \sum_{i,j\;:\;x_i \geq x_j} P_{ij} \int_{x_j}^{x_i} t \;dt = 4 \int_{0}^{\infty} t \sum_{ij\;:\; x_i > t \geq x_j} P_{ij}\; dt,
\ee
and furthermore,
 \be
 \sum_{ij\;:\; x_i > t \geq x_j} P_{ij} = \sum_{i \in A(t)}\sum_{j \in A^c(t)} P_{ij} \;\;\; \mbox{where,} \; \; \; A(t) \equiv \left\{ i | x_i \geq t\right\},
 \ee
 we can bound, 
 \be
 4 \int_{0}^{\infty} t \sum_{ij\;:\; x_i > t \geq x_j} P_{ij}\; dt \geq h\; 4\int_0^{\infty}t \sum_{i \in A(t)} \Theta(t-x_i) \; dt = 2 \; h \left(\sum_i x_i^2\right),
 \ee
where we defined $h$ as in the same fashion as above. We have therefore found the desired lower bound for the spectral gap of the map $S$. \qed
\end{proof}

%
%subsection single-qubit example
\subsection{Example: Conductance bound for unital qubit channels}
A convenient basis for the matrix space $\M_2$  associated with the Hilbert space  $\HS = \cc^2$ is given in terms of the Pauli basis
$\{\1,\sigma_x,\sigma_y , \sigma_z\}$. In this basis a density matrix $\rho \in \sym_2$ can be parametrized in terms of its Bloch vector $ {\bf r} \in \rr^3$. 
In the Bloch representation  the density matrix reads $\rho = \frac{1}{2}\left( \1 + {\bf r}  \cdot {\bf \Sigma}\right)$, 
where ${\bf \Sigma} = ( \sigma_x , \sigma_y, \sigma_z)$. It is also straight forward to determine the matrix representation of a 
quantum channel $T : \M_2 \rightarrow \M_2$ with respect to the Pauli basis. A general channel can be written as a matrix $\hat{T} \in \M_{4}$. 
\bea
	\hat{T} = \left(\begin{array}{cc}  1 & 0 \\ {\bf t}  &  {\bf L} \end{array}\right). 
\eea
 The channel acts on a density matrix via $ T(\rho) = T(\frac{1}{2}\left( \1 + {\bf r}  \cdot {\bf \Sigma}\right)) = 
 \frac{1}{2}\left( \1 + ({\bf t} + {\bf L}{\bf r})  \cdot {\bf \Sigma}\right)$.  It can be shown, that the map $T$ is unital if and only if ${\bf t} = 0$. 
 Let us now consider the optimization for Cheeger's constant $h$ as given in Lemma (\ref{Lemma:conductanceBound}). Given the constraint, 
 we have to vary all one dimensional projectors  $\Pi_A =  \ket{\psi}\bra{\psi}$ with $\| \ket{\psi} \|_2 = 1$, so that 
 \bea
 		h = \min_{\ket{\psi} \in \cc^2} \mbox{tr}\left[  \left (\1 - \ket{\psi}\bra{\psi} \right) S \left(\ket{\psi}\bra{\psi} \right )\right].
 \eea  
The symmetrized map $S$ of the unital channel $T$, with  ${\bf t} = 0$, now assumes the matrix representation,
\bea
	\hat{S} = \left(\begin{array}{cc}  1 & 0 \\ 0 &  {\bf L}^\dagger {\bf L} \end{array}\right).
\eea 
Furthermore note, that any projector $ \ket{\psi}\bra{\psi} \in \sym_2$ can be parametrized via a Bloch vector ${\bf a} \in \rr^3$ that obeys 
$\| {\bf a} \|_2 = 1$. The minimization for Cheeger's constant reduces therefore to
\bea
	h = \min_{\| {\bf a} \|_2 = 1} 1 - \bra{\bf a} \; {\bf L}^\dagger {\bf L} \;  \ket{{\bf a}},  
\eea   
where $\langle{\bf a} \ket{{\bf b}}$ denotes the canonical scalar product in $\rr^3$. The minimum is attained when ${\bf a}$ is the eigenvector 
associated with the largest eigenvalue $s_1^2$ of the matrix $ {\bf L}^\dagger {\bf L}$. 
Hence for an arbitrary single qubit unital channel, Cheeger's constant  is given by $h = 1 - s_1^2$, where $s_1$ is the largest singular value
of the matrix ${\bf L}$ and hence the second largest singular value of the channel $T$.  We see that the conductance bound as stated in 
(\ref{Lemma:conductanceBound}) is indeed satisfied, since 
\bea
	2s_1^2 -1  \leq s_1^2 \leq \frac{1}{2}(1+ s_1^2).
\eea

%
% Conclusion and Discussion
%
\section{Conclusion and Discussion}\label{concl}
We have seen that by generalizing the $\chi^2$-divergence to the quantum setting, many of the classical results for the convergence of Markov processes can be recovered. The general perception, that the convergence should be governed by the spectral properties of the quantum channel could be verified in the asymptotic limit. The fact that we were working with non-commuting probabilities gave rise to a larger set of possibilities of defining an inversion of the fixed point density matrix, all of which give rise to a valid upper bound for the trace distance. An interesting question is how the different singular values $s^k_i$ of the corresponding quantum discriminant 
relate to each other.  The generalization of the $\chi^2$-divergence also led to the definition of detailed balance for quantum channels. Again, 
not only a single condition for quantum detailed balance exists but an entire family of conditions each determined by a different function $k \in \K$, all of which 
coincide in the case when we consider classical stochastic processes on a commuting subspace. The quantum concept of detailed balance therefore 
appears to be richer and allows for a wider set of channels to obey this definition. The conductance bound that was derived could only be shown 
for unital quantum channels. However we would like to point out,  that it is possible to give conductance bounds for classical maps when the Markov chain is not 
doubly stochastic. The fact that in general we may not assume that the fixed point of an arbitrary channel commutes with the eigenvector associated to the 
second largest eigenvalues seems to hinder a generalization for non-unital channels. Moreover, the classical conductance bound has a nice geometrical interpretation 
in terms of the cut set analysis and the maximal flow on the graph associated to the stochastic matrix $P_{ij}$. When investigating general quantum channels 
such a nice geometric interpretation seems to be lacking. For unital quantum channels Cheeger's constant can also be viewed in terms of the minimal probability 
flow of one subspace to its compliment.   

\section{Acknowledgements}
Part of this work was done during a workshop at the Erwin Schr\"odinger Institute for Mathematical Physics.  MMW and MJK thank David Reeb for helpful discussions.
KT was supported by the FWF doctoral program CoQuS. FV is supported by the FWF grant FoQuS and the European 
grants QUEVADIS and QUERG. MJK was supported in part by the NBIA. MMW acknowledges financial support by the Danish Research Counsil (FNU), QUANTOP and the European projects QUEVADIS and COQUIT.   MBR was partially supported by US National Science Foiundation Grant DMS-0604900.

\appendix 

\section{Families of divergences} \label{appA}

The most widely used family of divergences, often called $\alpha$-divergence \cite[Chapter 7]{AN},
is associated with the functions 
\bea
k_\alpha^{\rm WYD} (w) = \frac{(1 - w^\alpha)(1 - w^{1- \alpha })}{\alpha(1-\alpha) \, (1-w)^2 }
\qquad \hbox{for}  \qquad 
 \alpha \in [-1,2]
\eea 
This family is sometimes called the WYD divergences, because it arises
from an extension of the Wigner-Yanase-Dyson entropy \cite{Lb,LR} associated with the (unsymmetrized) function
$g(w) = \frac{1}{\alpha(1-\alpha)   } (w - w^\alpha)$.   In the limit $\alpha \rightarrow 1 $ this
yields \cite{JR}  the familiar (asymmetric) relative entropy 
$H(\rho, \sigma) = {\rm Tr} \rho (\log \rho - \log \sigma)$ and 
 $\Omega^{\log}_P$ given by  \eqref{omegalog}. Like the  family of divergences introduced here,
the minimal WYD divergence occurs for $\alpha = 1/2$, it is convex in $\alpha$, symmetric around $\alpha = 1/2$
and yields the maximal $\frac{1+w}{2w}$ when $\alpha = -1$ or $2$.   However,
$\alpha = 1/2 $ gives  $ k_{1/2}^{\rm WYD} (w) = 4 (1 + \sqrt{w})^{-2}$ which is quite different from
$k^{\rm mean}_{1/2}(w) = \sqrt{w} $.
 The WYD  family is often studied only for
  $\alpha \in (0,1)$; it was first observed by Hasegawa in \cite{Hasegawa}   that it yields a
monotone metric if and only if   $\alpha \in [-1,2]$.  

The metrics associated with    $k^{\rm mean}_\alpha(w)$ and $k_\alpha^{\rm WYD} (w)$
both give increasing families for $\alpha \geq \frac{1}{2} $ and both yield
 the maximal metric $k(w) =(1+w)(2w)$ for $\alpha$ the maximal values of $1$ and $2$
 respectively.  However, neither reduces to the minimal metric 
 $k(w) = 2/(1+w) $.   The measure $\delta(s - a) $ in \eqref{intOmega} leads to the family 
 $k_a(w) =  \frac{ (1+a)^2}{2} \frac{ (1+w)}{ (1 + wa )( w + a) }$  for $a \in [0,1]$
 which reduces to the 
 the maximal and minimal functions for $a = 0,1$.  However, this family is
 neither increasing nor decreasing.  
  Hansen \cite{Hansen} has found families of functions which increase
 monotonically from the smallest to the largest of which we mention
 only  
 \bea
    k_a(w) = w^{-a}  \Big( \frac{1 + w}{2} \Big)^{2a-1}  \qquad \hbox{for}  \qquad 
 a \in [0,1] ~.
 \eea

\bigskip

\section{Proof of a key inequality}
The proof of the contractivity of a general Riemannian metric is based on the following theorem first proved in \cite{LR}.

\begin{theorem}\label{proofMono}
For a channel $T : \M_d \rightarrow  \M_d $, we have that,
\bea\label{inequRuskai}
\tr\left[A^{\dagger}\frac{1}{R_\sigma + sL_\rho} A\right] = \tr\left[T \left(A^{\dagger}\frac{1}{R_\sigma + sL_\rho}A\right)\right] \geq \\\nonumber
\tr\left[T(A)^{\dagger}\frac{1}{R_{T(\sigma)} + sL_{T(\rho)}} T(A)\right].
\eea
\end{theorem} 
\proof{
Let $\sigma > 0$, then $\tr[ A^{\dagger} \sigma A] \geq 0$, and $\tr[ A^{\dagger}A \sigma] \geq 0$ so that $L_\sigma$ as well as 
$R_\sigma$ are both positive semi definite super operators on the matrix space. Therefore we infer, that for a positive $\rho > 0 $ 
the operator $R_{\sigma} + s L_{\rho}$ is also positive. We define a matrix $X = [R_\sigma + sL_\rho]^{-1/2}(A) + [R_\sigma + sL_\rho]^{1/2}T^*(A)$ 
and furthermore $B = [R_{T(\sigma)} + sL_{T(\rho)}]^{-1}T(A)$. Since $\tr[X^{\dagger}X] \geq 0$, we have that
 \be\label{messy}
  \tr\left[A^\dagger\frac{1}{R_\sigma + s L_\rho}A \right] - \tr\left[T^*(B^\dagger)A\right]  
-  \tr\left[A^\dagger T^*(B)\right] + \tr\left[T^*(B^\dagger)[R_\sigma + s L_\rho]T^*(B)\right] \geq 0.
\ee 
Furthermore note, that 
\be
 -\tr\left[A^{\dagger} T^*(B) \right] - \tr\left[T^*(B^\dagger) A \right] = -2 \tr\left[T(A^\dagger) \frac{1}{R_{T(\sigma)} + s L_{T(\rho)}} T(A) \right].
\ee
It therefore suffices to show that we are able to bound the last term in (\ref{messy}) by the right side of the inequality (\ref{inequRuskai}).
Note, that 
\bea
	\tr\left[T^*(B^\dagger)[R_\sigma + s L_\rho]T^*(B)\right] = \tr\left[T^*(B^\dagger)T^*(B)\sigma + sT^*(B^\dagger)\rho T^*(B) \right] \\\nonumber
	\leq \tr\left[T^*(B^\dagger B)\sigma + sT^*(B B^\dagger)\rho\right], 
\eea
since $\rho,\sigma > 0$ and due to the operator inequality $T^*(B^\dagger)T^*(B) \leq T^*(B^\dagger B)$. This inequality holds for 
any $B$ since $T$ is a channel and by that trace preserving, hence $T^{*}(\1) = \1$. With 
$\tr\left[ T^*(B^\dagger B) \sigma \right] = \tr\left[ B^\dagger B T(\sigma) \right]$ we can write 
\bea
	&&\tr\left[T^*(B^\dagger)[R_\sigma + s L_\rho]T^*(B)\right] \leq \tr\left[B^\dagger BT(\sigma) + s B^\dagger BT(\rho)\right] \\\nonumber
	&=& \tr\left[B^\dagger [R_{T(\sigma)} + s L_{T(\rho)}] B\right] =  \tr\left[B^\dagger T(A)\right] =
	  \tr\left [T(A^{\dagger })\frac{1}{R_{T(\sigma)} + s L_{T(\rho)}} T(A) \right]. 
\eea
\qed}


\begin{thebibliography}{34}
\bibitem{Pears:chi2} K. Pearson, On the Criterion that a given system of deviations from the 
probable in the case of correlated system of variables is such that it can be reasonable supposed to have arisen from random sampling, Philos. 
Mag. 50 (1900) 157--172.
\end{thebibliography}
\end{document}